\documentclass[twocolumn,prl,aps]{revtex4}

\usepackage{graphicx}
\usepackage{dcolumn}
\usepackage{bm}
\usepackage{amsmath}
\usepackage{amsthm}
\usepackage{verbatim}
\usepackage{amssymb}

\newtheorem{theorem}{Theorem}
\newtheorem{lemma}{Lemma}
\newtheorem{proposition}{Proposition}

\newtheorem{definition}{Definition}
\newtheorem{remark}{Remark}

\newcommand*{\cH}{\mathcal{H}}

\newcommand*{\cM}{\mathcal{M}}
\newcommand*{\cN}{\mathcal{N}}

\newcommand*{\mM}{\mathbb{M}}

\newcommand{\ket}[1]{|#1\rangle}
\newcommand{\bra}[1]{\langle #1 |}

\newcommand{\proj}[1]{\ket{#1}\bra{#1}}

\newcommand{\ot}[0]{\otimes}
\newcommand{\beq}{\begin{equation}}
\newcommand{\eeq}{\end{equation}}

\newcommand{\Tr}{{\rm Tr}}

\newcommand{\rhoAB}{\rho_{AB}}

\newcommand{\setft}[1]{\mathrm{#1}}

\newcommand{\density}[1]{\setft{D}\left(#1\right)}
\newcommand{\dsep}{\setft{D}_{\rm SEP}}

\newcommand{\dppt}{\setft{D}_{\rm PPT}}

\begin{document}

\title{Relative Entropy of Entanglement and Restricted Measurements}

\author{M. Piani}
\affiliation{Institute for Quantum Computing \& Department of Physics
  and Astronomy, University of Waterloo, 200 University Ave. W., N2L 3G1 Waterloo, Canada}


\begin{abstract}
 We introduce variants of relative entropy of entanglement based on the optimal distinguishability from unentangled states by means of restricted measurements. In this way, we are able to prove that the standard regularized entropy of entanglement is strictly positive for all multipartite entangled states. In particular, this implies that the asymptotic creation of a multipartite entangled state by means of local operations and classical communication always requires the consumption of a non-local resource at a strictly positive rate.
\end{abstract}


\maketitle

Entanglement is often considered the fundamental trait of quantum mechanics. Besides its conceptual relevance, it plays a crucial role in quantum information processing, as it lies, e.g., at the core of tasks like quantum teleportation~\cite{BennettBCJPW93} or dense coding~\cite{BennettW92}.
Given its usefulness, one of the major issues in entanglement theory~\cite{HHHH07} is its quantification. Many different measures of entanglement have been introduced in the years. They have proved to be
useful mathematical and conceptual tools, with links to many aspects of quantum information processing~\cite{HHHH07,PV07}.

While there is in principle an infinite number of such possible measures, some of them can be considered, for different reasons, worth special attention.
For example there are measures with operational meaning in the so-called distant-labs paradigm, also known as Local Operations and Classical Communication (LOCC) paradigm, where entanglement is elevated to the status of precious resource by imposing constraints to the kind of operations that spatially separated parties can perform. In the bipartite setting, this is the case for entanglement cost $E_c$ and distillable entanglement $E_d$. $E_c$ is the rate at which pre-
A second family is given by those measures that have some sort of geometrical interpretation, being based on the notion of distance
from the set of separable states~\cite{relent2}. As such they are especially useful, as they constitute the quantitative correspondent to a qualitative and intuitive reasoning based on the structure of the set of states. Among the last family, relative entropy of entanglement~\cite{relent1,relent2} stands out as an elegant and powerful tool in entanglement theory, being based on relative entropy, a fundamental quantity in quantum information theory~\cite{relentSW,relentV}.

The quantum relative entropy of a state $\rho$ with respect to a state $\sigma$ is defined as
$S(\rho||\sigma)\equiv\Tr \rho \log \rho - \Tr \rho \log \sigma$~\footnote{Relative entropy is infinite if the support of $\rho$ is not included in the support of $\sigma$.}.
It finds its operational interpretation in asymmetric quantum hypothesis testing. When given two hypothesis to test ---via measurements--- in the form of $n$ i.i.d. copies of either $\rho$ or $\sigma$, $S(\rho||\sigma)$ corresponds to the optimal rate of decay with $n$ of the probability of error of finding that the state was $\rho$ when it was actually $\sigma$~\footnote{This can be achieved when asking only that the converse error---finding that the state was $\sigma$ when it was actually $\rho$---be bounded.}. This is the content of Quantum Stein's Lemma~\cite{stein1,stein2}.

The relative entropy of entanglement of a bipartite state $\rhoAB$ is defined as $E_R(\rhoAB)\equiv\min_{\sigma_{AB}}S(\rhoAB\|\sigma_{AB})$, where the minimum runs over all separable states $\sigma_{AB}=\sum_ip_i\sigma^i_A\ot\tilde{\sigma}^i_B$. It is a faithful entanglement measure, in the sense that it is strictly positive if and only if $\rhoAB$ is entangled, i.e. non-separable. In many cases---for example when considering the relation of $E_c$ and $E_d$ with other suitable entanglement measures~\cite{liment}---one deals with the asymptotic regularization---henceforth called standard regularization---of a function on states $f$, defined by $f^{\infty}(\rho)\equiv\lim_n \frac{1}{n}f(\rho^{\ot n})$. Brand\~ao and Plenio recently provided a operational meaning to $E_R^\infty$, which is known to satisfy $E_c\geq E_R^\infty\geq E_d$: it is the rate of reversible manipulation of entanglement by means of class of (asymptotically) non-entangling operations~\cite{reverseent}, that is, $n$ copies of a state $\rho$ can be reversibly transformed into $\approx nE_R^\infty(\rho)/E_R^\infty(\sigma)$ copies of a state $\sigma$. It was known for a long time that $E_R$ can be strictly subadditive, i.e. there exist states $\rhoAB$ such that $E_R(\rhoAB^{\otimes 2})<2E_R(\rhoAB)$~\cite{subaddrelent}. The problem is that states $\sigma_{AA':BB'}$ that are separable in the $AA':BB'$ bipartite cut, and with which we compare states $\rhoAB\otimes\tau_{A'B'}$ to compute $E_R(\rhoAB\otimes\tau_{A'B'})$, may actually be correlated or even entangled in the cut $AB:A'B'$. As a consequence, $E_R^\infty\neq E_R$ in general and it was not known whether $E_R^\infty$ was faithful. Thus, it could be that, in the reversible theory of entanglement of Brand\~ao and Plenio, the asymptotic conversion of one entangled state $\rho$ to another entangled state $\sigma$ does not require a non-zero rate of consumption of the former state. Similarly, while in the bipartite setting it is know that entanglement cost is faithful~\cite{yang05}, in the multipartite setting it was not known  whether the asymptotic transformation via LOCC of $\rho$ to $\sigma$ always requires the consumption of copies of $\rho$ at a finite rate. As non-entangling operations are a larger class than LOCC, if $E_R^\infty$ is faithful then also for LOCC the rate must be strictly positive for all entangled states.

Here we derive a simple but powerful inequality satisfied by $E_R$, using the notion of restricted---e.g., to LOCC---measurements, and we are naturally led to define a new version of relative entropy of entanglement based on such restricted measurements. As a corollary, we prove that $E_R^\infty$ is faithful. Recently, a similar result as been independently established by Brand\~ao and Plenio, who obtained it as a non-trivial corollary of a generalization of Stein's Lemma~\cite{brandao09}. Our approach has the the advantage of simplicity and appears to be of wide applicability in entanglement theory and quantum information. 

In order to state our result we will need to establish some notation and some definitions. The set of positive operators of trace one---states---on a Hilbert space $\cH$ will be denoted $\density{\cH}$. The relative entropy of a probability distribution $(p_i)_i$ with respect to a probability distribution $(q_i)_i$ is defined as
$S((p_i)_i||(q_i)_i)\equiv \sum_i p_i \log\frac{p_i}{q_i}$~\footnote{It is infinite if $p_i>0$ for some $i$ such that $q_i=0$.}. Useful properties of quantum relative entropy are listed in the following proposition.
\begin{proposition}
\label{prop:relentr}
The quantum relative entropy satisfies:
\begin{enumerate}
\item $S(\rho\|\sigma)\geq 0$ and $S(\rho\|\sigma)=0\Leftrightarrow \rho=\sigma$;
\item $\sum_i p_i S(\rho_i\|\sigma_i)\geq S(\sum_i p_i\rho_i\|\sum_i p_i\sigma_i)$;
\item $S(\rho\|\sigma)\geq S(\Lambda(\rho)\|\Lambda(\sigma))$, for any completely positive trace-preserving map $\Lambda$.
\end{enumerate}
\end{proposition}
\begin{definition}
\label{def:measop}
A \emph{measurement operation} $\cM$ is associated to a POVM measurement $(M_i)_i$, $M_i\geq 0$, $\sum_i M_i=\openone$ via
\[
\mathcal{M}(X)=\sum_ip_i(X)\proj{i}, \quad p_i(X)=\Tr(M_iX)
\]
with $\{\ket{i}\}$ an orthonormal set.
\end{definition}
Notice that
\[
S(\rho\|\sigma)\geq S(\mathcal{M}(\rho)\|\mathcal{M}(\sigma))=S((p_i(\rho))_i\|(p_i(\sigma))_i),
\]
where the first inequality comes from property 3 of Proposition \ref{prop:relentr}, and the particular choice for the orthonormal set $\{\ket{i}\}$ in Definition \ref{def:measop} is irrelevant. In the following we will often use interchangeably the words ``measurement'', ``POVM'' and ``measurement map''.

A measurement may be arbitrary or be restricted to a particular class of measurements $\mM$, and we may indicate that by writing, with an abuse of notation, $\cM\in\mM$ for the corresponding measurement map. Following \cite{matthews08}, let us consider a multi-partite system with $n$ parties. The total Hilbert space is then $\cH=\cH_1\otimes\cdots\otimes\cH_n$, with $\cH_j$ a local Hilbert space of dimension $d_j$. For example, in such a setting one may consider the following classes of restricted measurements:
local measurements $\mM_{\rm LO}$ for which $M_i=M_{k_1}^{(1)}\ot\cdots\ot M_{k_n}^{(n)}$ with each $(M_{k_j}^{(j)})_{(k_j)}$ a POVM on $\cH_j$; LOCC measurements $\mM_{\rm LOCC}$ (of involved characterization); separable measurements $\mM_{\rm SEP}$ for which $M_i=\sum_{k}M_{i,k}^{(1)}\ot\cdots\ot M_{i,k}^{(n)}$, for $M_{i,k}^{(j)}\geq0$; positive under partial transposition (PPT) measurements $\mM_{\rm PPT}$ for which every $M_i$ is PPT with respect to every possible bipartition. It holds $\mM_{\rm LO}\subset\mM_{\rm LOCC}\subset\mM_{\rm SEP}\subset\mM_{\rm PPT}$.
\begin{definition}
The quantum relative entropy of $\rho\in\density{\cH}$ with respect to $\sigma\in\density{\cH}$ and a class of measurement operations $\mathbb{M}$, or ${\mathbb{M}}$-relative entropy of $\rho$ with respect to $\sigma$, is defined as
\beq
\mathbb{M}S(\rho||\sigma)\equiv\sup_{\mathcal{M}\in\mathbb{M}}S(\mathcal{M}(\rho)||\mathcal{M}(\sigma))
\eeq
\end{definition}
Because of property 3 of Proposition \ref{prop:relentr}, $\mathbb{M}S(\rho||\sigma)\leq S(\rho||\sigma)$, but if the measurements are unconstrained, it is known that $\lim_n \frac{1}{n}\mathbb{M}S(\rho^{\otimes n}||\sigma^{\otimes n})=S(\rho\|\sigma)$~\cite{stein1}.
\begin{remark}
If $\mathbb{M}$ contains informationally complete measurements~\cite{CFS02}, i.e. any measurement $\mathcal{M}$ such that $\mathcal{M}(\rho)=\mathcal{M}(\sigma)$ if an only if $\rho=\sigma$, then $\mathbb{M}S(\rho||\sigma)=0$ if and only if $\rho=\sigma$. The set $\mathbb{M}_{\rm LO}$ contains informationally complete measurements, and so do all the others $\mM=\mathbb{M}_{\rm LOCC},\mathbb{M}_{\rm{SEP}},\mM_{\rm PPT}$.
\end{remark}
\begin{definition}
Given $\rho\in\density{\cH}$ and a reference set $P\subset\density{\cH}$, the relative entropy of $\rho$ with respect to $P$ is defined as
\beq
E^P_R(\rho)\equiv\inf_{\sigma\in P}S(\rho||\sigma),
\eeq
and the ${\mathbb{M}}$-relative entropy of $\rho$ with respect to $P$ is defined as
\beq
\mathbb{M}E^P_R(\rho)\equiv\inf_{\sigma\in P}\mathbb{M}S(\rho||\sigma).
\eeq 
\end{definition}
Because of property 3 of Proposition \ref{prop:relentr}, $\mathbb{M}E^P_R(\rho)\leq E^P_R(\rho)$.
We will always consider reference sets $P$ which are convex and compact, like the subset of separable states $\dsep(\cH)=\{\sigma=\sum_ip_i\sigma_i^{(1)}\otimes\cdots\otimes\sigma_i^{(n)}\}$ or the subset $\dppt$ of states that are PPT with respect to every possible bipartition. Because of property 2 of Proposition \ref{prop:relentr}, in such a case there exist an optimal reference state $\sigma^*\in P$ (depending on $\rho$) such that $E^P_R(\rho)=S(\rho||\sigma^*)$ and a---potentially different---optimal reference state $\sigma_\mM^*\in P$ such that $\mathbb{M}E^P_R(\rho)=\mathbb{M}S(\rho||\sigma_\mM^*)$.
\begin{remark}
$E^P_R(\rho)=0$ if and only if $\rho\in P$. Moreover, if $\mathbb{M}$ contains informationally complete measurements, then also $\mathbb{M}E^P_R(\rho)=0$ if and only if $\rho\in P$. 
\end{remark}
With an abuse of notation, by $P$ we will from now on indicate a family of reference sets rather than a single reference set. For example, we may take $P$ to be the family of bipartite separable states, with local parties not having a definite dimension~\footnote{This is not in contradiction with the fact that for fixed dimensions of subsystems, $P$ is convex and compact.}. Thus, $\sigma_{AA'BB'}=\sum_ip_i \sigma_{AA'}^i\otimes\tilde{\sigma}_{BB'}^i$ is separable with respect to the bipartite cut $AA':BB'$, and its reduced state $\sigma_{AB}=\Tr_{A'B'}(\sigma_{AA'BB'})=\sum_ip_i \sigma_{A}^i\otimes\tilde{\sigma}_{B}^i$ is also separable, with respect to the $A:B$ cut, and we will say that they are both in $P$. If we denote $X=AB$ and $Y=A'B'$, we may write that both $\sigma_{XY}$ and $\sigma_X$ are in $P$.

We are now ready to state our main result. 
\begin{theorem}
\label{thm:main}
Consider two systems $X$ and $Y$ with joint Hilbert space $\cH_X\ot\cH_Y$, and a convex reference set $P$. Suppose that the set $\mathbb{M}$ of measurement operations on $X$ and the reference set $P$ are such that for all POVM elements $M_i$ associated to any measurement, and all $\sigma_{XY}\in P$, $\mathcal{M}\in\mathbb{M}$ on $X$, $\Tr_{X}(M_i^X\sigma_{XY})\in P$ (up to normalization). Further, suppose that $P$ is closed under partial trace, so that in particular $\sigma_X\in P$. Then, for any $\rho_{XY}\in\density{\cH_X\ot\cH_Y}$,
\beq
\label{eq:main}
E^P_R(\rho_{XY})\geq \mathbb{M}E^P_R(\rho_X)+E^P_R(\rho_Y),
\eeq
with $\rho_{X}=\Tr_Y(\rho_{XY})$, and $\rho_{Y}=\Tr_Y(\rho_{XY})$.
\end{theorem}
Before proving Theorem 1, let us observe that inequality \eqref{eq:main} implies by recursion that $E^P_R(\rho_X^{\ot n})\geq n \mM E^P_R(\rho_X)$, so that also
\[
(E^P_R)^\infty(\rho_X)\geq \mM E^P_R(\rho_X) 
\]
If $\mM$, besides satisfying the hypothesis of the theorem, contains informationally complete measurements, we find $(E^P_R)^\infty(\rho_X)>0$ for all $\rho\notin P$.






In order to prove the theorem we will need the following, easily checked observation.
\begin{lemma}
\label{lem:basic}
 Given two ensembles $\{(r_k,\rho_k)\}$ and $\{(s_k,\sigma_k)\}$, with $(r_k)_k$ and $(s_k)_k$ probability distributions, and an orthonormal basis $\{\ket{k}\}$, one has
\begin{multline}
S\Big(\sum_kr_k\rho_k\otimes\proj{k}\Big\|\sum_ks_k\sigma_k\otimes\proj{k}\Big)\\
=S((r_k)_k||(s_k)_k)+\sum_kr_kS(\rho_k||\sigma_k)
\end{multline}
\end{lemma}




\begin{proof}[Proof of Theorem \ref{thm:main}]
Given two states $\rho_{XY}$ and $\sigma_{XY}\in P$, for all measurement maps $\mathcal{M}_X$ on $X$ we have
\begin{multline}
\label{eq:proofmain}
S(\rho_{XY}\|\sigma_{XY})\\
\begin{aligned}
	&\stackrel{(i)}{\geq} S(\mathcal{M}_X[\rho_{XY}]\|\mathcal{M}_X[\sigma_{XY}])\\
	&\stackrel{(ii)}{=}S(\sum_ip_i(\rho_X) \proj{i}\otimes\rho_Y^i \| \sum_ip_i(\sigma_X)\proj{i}\otimes \sigma_Y^i )\\
	&\stackrel{(iii)}{=}S((p_i(\rho_X))\|(p_i(\sigma_X)))+\sum_i p_i(\rho_X) S(\rho_Y^i\|\sigma_Y^i)\\
	&\stackrel{(iv)}{\geq} S(\cM_X(\rho_X)\|\cM_X(\sigma_X))\\
	&\qquad+S(\sum_i p_i(\rho_X)\rho_Y^i\|\sum_i p_i(\rho_X)\sigma_Y^i)\\
	&\stackrel{(v)}{=}S(\cM_X(\rho_X)\|\cM_X(\sigma_X)) + S(\rho_Y\|\sum_i p_i(\rho_X)\sigma_Y^i)\\
	&\stackrel{(vi)}{\geq}S(\cM_X(\rho_X)\|\cM_X(\sigma_X)) + E_R^P(\rho_Y)
\end{aligned}
\end{multline}
where we used: in (i), property 3 of Proposition \ref{prop:relentr}; in (ii), the definition of measurement map, with $p_i(\rho_X)=\Tr_X(M_i^X\rho_X)=\Tr_{XY}(M_i^X\otimes\openone_Y \rho_{XY})$ and the conditional states $\rho_Y^i=\Tr_X(M_i^X\otimes\openone_Y \rho_{XY})/p_i(\rho_X)$ (similarly for $\sigma$); in (iii), Lemma \ref{lem:basic}; in (iv), property 2 of Proposition \ref{prop:relentr}; in (v), that the measurement map is trace-preserving; in (vi), that by hypothesis $\sum_i p_i(\rho_X)\sigma_Y^i\in P$. As this is valid for any measurement, we obtain that for all $\sigma_{XY}$
\[
S(\rho_{XY}\|\sigma_{XY})\geq \mM S(\rho_X\|\sigma_X)+E_R^P(\rho_Y)
\]
By assumption, $\sigma_{X}\in P$, therefore,
\begin{align*}
E_R^P(\rho_{XY})&=\inf_{\sigma_{XY}\in P}S(\rho_{XY}\|\sigma_{XY})\\
		&\geq \inf_{\tilde{\sigma}_{X}\in P}\mM S(\rho_X\|\tilde{\sigma}_X)+E_R^P(\rho_Y)\\
		&= \mM E_R^P(\rho_{XY}) + E_R^P(\rho_Y)
\end{align*}
\end{proof}


It is straightforward to check that Theorem \ref{thm:main} applies in particular to any combination of the cases $P=\dsep,\dppt$ and $\mathbb{M}=\mathbb{M}_{\rm LO},\mathbb{M}_{\rm{LOCC}},\mathbb{M}_{\rm{SEP}}$. 
Further it applies to the case $P=\dppt$ and $\mM=\mM_{\rm PPT}$.

In order to obtain a more explicit lower bound, we observe that by Pinsker inequality
\[
S(\mathcal{M}(\rho)\|\mathcal{M}(\sigma))\geq \frac{1}{2\ln 2} (\|\mathcal{M}(\rho)-\mathcal{M}(\sigma)\|_1)^2
\]
with $\|A\|_1=\Tr\sqrt{A^\dag A}$ the trace norm.
According to \cite{matthews08},
\[
\sup_{\mathcal{M}\in\mathbb{M}_{\rm SEP}}\|\mathcal{M}(\rho)-\mathcal{M}(\sigma)\|_1\geq \frac{2}{2^{n/2}}\frac{1}{\sqrt{D}}\|\rho-\sigma\|_1
\]
with $n$ the number of parties and $D$ the total dimension $D=d_1d_2\ldots d_n$, thus we have for example
\[
\mathbb{M}_{\rm SEP}E^P_R(\rho)\geq \frac{1}{2^{n-1}D\ln2}(\inf_{\sigma\in P}\|\rho-\sigma\|_1)^2.
\]

We would like to remark that the result of Theorem \ref{thm:main} leads to interesting results other than the faithfulness of $E_R^\infty$. In \cite{condent} it was shown that from any entanglement measure $E$ for $n$-party entanglement, one can define a new one by means of conditioning, as:
\begin{multline*}
\label{eq:condent}
CE(\rho_{A_1A_2\ldots A_n})\equiv\inf_{\sigma}[E(\sigma_{A_1A'_1A_2A'_2\ldots A_nA'_n})\\
	-E(\sigma_{A'_1A'_2\ldots A'_n})]
\end{multline*}
where $A'_i$ are local ancillas of the systems $A_i$, and the infimum is over extensions $\sigma_{A_1A'_1A_2A'_2\ldots A_nA'_n}$ satisfying $\sigma_{A_1A_2\ldots A_n}=\rho_{A_1A_2\ldots A_n}$. One checks that $CE$ is naturally superadditive, i.e. $CE(\sigma_{A_1A'_1A_2A'_2\ldots A_nA'_n})\geq CE(\sigma_{A_1A_2\ldots A_n})+CE(\sigma_{A'_1A'_2\ldots A'_n})$.
Another concept recently developed~\cite{PCMH09} is that of broadcast---as opposed to ``standard''---regularization of a function $f$ on states. For a state $\rho\equiv\rho_X$ it is defined as $f^\infty_b(\rho)=\lim_m\frac{1}{m}\min_{\rho^{(m)}}f(\rho^{(m)})$, with $\rho^{(m)}\equiv\rho^{(m)}_{X^m}$, $X^m\equiv X_1\ldots X_m$, a $m$-copy broadcast state of $\rho$, i.e. $\rho^{(M)}_{X_k}\equiv {\rm Tr}_{X_1 \cdots X_{k-1} X_{k+1} \cdots X_m}\rho^{(m)}=\rho$ for all $k$. One readily verifies that
\beq
\label{eq:ineq}
(E^P_R)^\infty(\rho_{X})\geq (E^P_R)_b^\infty(\rho_{X})\geq CE^P_R(\rho_{X})\geq \mathbb{M}E^P_R(\rho_{X}),
\eeq
where: the first inequality is due to the fact that broadcast copies are a particular case of i.i.d. copies; the second inequality comes from the iterative use of the broadcasting condition and to the fact that any broadcast copy is a particular extension of the single copy state; the last inequality is due to Theorem \ref{thm:main}.

We further notice that for $\rho\equiv\rho_{A_1\ldots A_n}$, the multipartite mutual information $I(A_1:\ldots:A_n)_{\rho}\equiv S(\rho\|\rho_{A_1}\ot \cdots \ot\rho_{A_n})$ satisfies
\beq
\label{eq:mutual}
I(A_1:\ldots:A_n)_{\rho}\geq \min_{\sigma\in \dsep}S(\rho\|\sigma)= E^{\dsep}_R(\rho),
\eeq
because $\rho_{A_1}\ot \cdots \ot\rho_{A_n}$ is a particular separable state. Taking the broadcast regularization of the leftmost and rightmost terms of \eqref{eq:mutual}, we get $I_b^\infty\geq (E^{\dsep}_R)_b^\infty(\rho)$. Then, for $\mM=\mM_{\rm SEP}$ the inequalities \eqref{eq:ineq} provide a better lower bound $\mathbb{M}_{\rm SEP}E^{\dsep}_R$ for the asymptotic broadcast mutual information $I_b^\infty$ than the one exhibited in~\cite{PCMH09}. As argued in~\cite{PCMH09}, $I_b^\infty$ has many properties of and entanglement measure, and its strict positivity for all entangled states may be interpreted as a kind of monogamy of quantum correlations among the broadcast copies.


Finally, for suitable choices of $\mathbb{M}$ and $P$ we prove that $\mathbb{M}E^{P}_R$ is an entanglement measure itself. In particular this holds for $\mM=\mM_{\rm LOCC},\mM_{\rm SEP}$ and for $P=\dsep,\dppt$, on which we will focus for the sake of clarity and concreteness. The most striking feature of such generalizations of relative entropy of entanglement is that, while the latter is subadditive, they are superadditive entanglement measures. The proof of the following properties is presented in the Appendix, in particular the proof of superadditivity is similar to that of Theorem~\ref{thm:main}.
\begin{theorem}
\label{thm:entmeas}
For any combination of $\mM=\mM_{\rm LOCC},\mM_{\rm SEP}$ and $P=\dsep,\dppt$, $\mathbb{M}E^{P}_R$ is an entanglement measure which is:
(a) faithful; (b) convex; (c) strongly LOCC monotone: $\mathbb{M}E^{P}_R(\rho^{in})\geq \sum_i p^{out}_i\mathbb{M}E^{P}_R(\rho^{out}_i)$, with $\rho^{out}_i$ the possible outputs---each with probability $p^{out}_i$---of an LOCC operation on $\rho^{in}$; (d) strongly superadditive: $\mathbb{M}E^{P}_R(\rho_{XY})\geq \mathbb{M}E^{P}_R(\rho_{X})+\mathbb{M}E^{P}_R(\rho_{Y})$.
\end{theorem}

In conclusion, we have introduced new variants of relative entropy of entanglement based on the optimal distinguishability from unentangled states by means of restricted measurements. On the one hand, these variants, for a proper class of measurements, have themselves the full status of entanglement measures, and they are faithful, i.e., they vanish for and only for separable states. On the other hand, the original relative entropy of entanglement can be shown to satisfy a kind of superadditivity inequality involving the newly introduced quantities. Such a relation appears to be a powerful tool in entanglement theory. For example, it leads to a very simple and straightforward proof that asymptotic relative entropy of entanglement is faithful, both in the bipartite and multipartite setting. This implies that the asymptotic creation of a multipartite entangled state by means of local operations and classical communication always requires the consumption of a non-local resource at a strictly positive rate.

MP thanks J. Watrous for helpful discussions,
and acknowledges support from NSERC QuantumWorks and Ontario Centres of Excellence. 


%
%

\section*{Appendix}
\begin{proof}[Proof of Theorem~\ref{thm:entmeas}] In the following, infima and maxima are always understood to be over the chosen sets $P$ and $\mM$, if not otherwise specified.

(a) Faithfulness was already proved.

(b) Choose optimal ${\sigma^*_{\mM,i}}$'s for $\rho_i$'s. Then
\begin{multline}
\mM E^{P}_R(\sum_ip_i\rho_i)\\
\begin{aligned}
&\leq\sup_\cM S(\cM(\sum_ip_i\rho_i)\|\cM(\sum_ip_i{\sigma^*_{\mM,i}}))\\
&\stackrel{(i)}\leq\sup_\cM \sum_ip_iS\left(\cM(\rho_i)\|\cM(\sigma^*_{\mM,i})\right)\leq \sum_ip_i \mM E^{P}_R(\rho_i),
\end{aligned}
\end{multline}
where in (i) we have used linearity of $\cM$ and property 2 of Proposition~\ref{prop:relentr}.

(c) Having proved convexity, according to~\cite{Horo05}, it is sufficient to check the invariance of $\mathbb{M}E^{P}_R$ under local unitaries, and the [FLAGS] condition $\mathbb{M}E^{P}_R(\sum_ip_i\rho_i\ot\proj{i}_{A'_k})=\sum_ip_i \mM E^{P}_R(\rho_i)$, where $\{\ket{i}\}$ is an orthonormal basis for a local ancilla of party $A_k$, for all ensembles $\{(p_i,\rho_i)\}$ and all $k=1,\ldots,n$. Invariance under local unitaries derives immediately from $\mM$ and $P$ being closed under local unitaries, that is if $(M_i)_i$ is a POVM in $\mathbb{M}$ and $\sigma$ a state in $P$, then also $\left((\bigotimes_{k=1}^n U^{(k)})M_i(\bigotimes_{k=1}^n {U^{(k)}})^\dagger\right)_i\in \mM$ and $(\bigotimes_{k=1}^n U^{(k)})\sigma(\bigotimes_{k=1}^n {U^{(k)}}^\dagger)\in P$.

As regards the [FLAGS] condition, the direction ``$\leq$'' comes from convexity and from 
$\mathbb{M}E^{P}_R(\rho\ot\proj{\psi}_{A'_k})=\mathbb{M}E^{P}_R(\rho)$, for any pure state $\ket{\psi}$. The latter equality is easily checked:
\begin{multline*}
\mathbb{M}E^{P}_R(\rho\ot\proj{\psi}_{A'_k})\\
\begin{aligned}
&\leq\inf_{\sigma'}\sup_\cM S(\cM(\rho\otimes\proj{\psi}\|\cM(\sigma'\otimes\proj{\psi})\\
&\leq\inf_{\sigma'}\sup_\cM S((\Tr(\bra{\psi}M_i\ket{\psi}\rho))_i\|(\Tr(\bra{\psi}M_i\ket{\psi}\sigma'))_i)\\
&\leq\inf_{\sigma'}\sup_\cM'S(\cM'(\rho)\|\cM'(\sigma'))=\mathbb{M}E^{P}_R(\rho),
\end{aligned}
\end{multline*}
as $(\bra{\psi}M_i\ket{\psi})_i$ is a POVM in $\mM$ if $(M_i)_i$ is. On the other hand,
\begin{multline*}
\mathbb{M}E^{P}_R(\rho\ot\proj{\psi}_{A'_k})\\
\begin{aligned}
&\geq\inf_{\sigma}\sup_{(M_i\otimes\openone_{A'_k})_i} S((\Tr(M_i\rho))_i\|(\Tr(M_i\bra{\psi}\sigma\ket{\psi}))_i)\\
&\geq\inf_{\sigma}\sup_{\cM} S(\cM(\rho)\|\cM(\sigma))=\mathbb{M}E^{P}_R(\rho).
\end{aligned}
\end{multline*}

The direction ``$\geq$'' of [FLAGS] is proved by:
\begin{multline}
 \mM E^{P}_R(\sum_ip_i\rho_i\ot\proj{i})\\
\begin{aligned}
&\stackrel{(i)}{\geq} \inf_\sigma \sup_{\{\cM_i\}} S(\sum_ip_i\cM_i(\rho_i)\ot\proj{i}\|\sum_iq_i\cM_i(\sigma_i)\ot\proj{i})\\
&\stackrel{(ii)}{=} \inf_\sigma \sup_{\{\cM_i\}} [ S((p_i)_i \|(q_i)_i) + \sum_i p_i S(\cM_i(\rho_i)\|\cM_i(\sigma_i))]\\
&\stackrel{(iii)}{\geq} \sum_i p_i \inf_{\sigma_i} \sup_{\cM_i} S(\cM_i(\rho_i)\|\cM_i(\sigma_i)) = \sum_i p_i \mM E_R(\rho_i),
\end{aligned}
\end{multline}
where: in (i), we restricted the measurement $\cM$ to a ``controlled-measurement'' of the form $\sum_i \cM_i\otimes\proj{i}\cdot\proj{i}$, with $\cM_i$ measurement maps, so that $\sigma_i=\Tr_{A'_k}(\proj{i}_{A'_k}\sigma)/q_i$, with $q_i=\Tr(\proj{i}_{A'_k}\sigma)$; in (ii), we have used Lemma~\ref{lem:basic}; in (iii), we have discarded a positive contribution, and minimized independently every term of the convex combination.

(d) For every $\rho_{XY}$ and every $\sigma_{XY}\in P$, it holds
\begin{multline}
\mM S(\rho_{XY}\|\sigma_{XY})\\
\begin{aligned}
 &\stackrel{(i)}{\geq}\sup_{\cM_{X},\cN_{Y}}S(\cM_{X}\ot\cN_{Y}(\rho_{XY})\|\cM_{X}\ot\cN_{Y}(\sigma_{XY}))\\
 &\stackrel{(ii)}{=} \sup_{\cN_{Y}}\Big\{\sup_{\cM_{X}}\Big[S((p_i(\rho_X))_i\|(p_i(\sigma_X))_i)\\
 &\qquad\qquad\quad+\sum_ip_i(\rho_X)S(\cN_Y(\rho_Y^i)\|\cN_Y(\sigma_Y^i))\Big]\Big\}\\
 &\geq\sup_{\cM_{X}}\Big[S(\cM_{X}(\rho_X)\|\cM_{X}(\sigma_X))\\
 &\qquad\qquad\quad+\sup_{\cN_{Y}}S(\cN_Y(\rho_Y)\|\cN_Y(\sum_ip_i(\rho_X)\sigma_Y^i))\Big]\\
 &\geq  \inf_{\sigma_X}\sup_{\cM_{X}}S(\cM_{X}(\rho_X)\|\cM_{X}(\sigma_X))\\
 &+\inf_{\sigma_Y}\sup_{\cN_{Y}}S(\cN_Y(\rho_Y)\|\cN_Y(\sigma_Y))\\
&\geq \mM E_R(\rho_{X}) + \mM E_R(\rho_{X}).
 \end{aligned}
 \end{multline}
The steps are very similar to those of \eqref{eq:proofmain}. Inequality (i) is due to the fact that factorized measurements $\cM_{X}\ot\cN_{Y}$ may be suboptimal for the sake of $\mM S(\rho_{XY}\|\sigma_{XY})$. In (ii), $p_i(\tau_X)=\Tr(M_X^i\tau_X)$ and $\tau_Y^i=\Tr_X(M_X^i\tau_X)/p_i(\tau_X)$, for $\tau=\rho,\sigma$, with $(M_X^i)_i$ the POVM corresponding to $\cM_X$. The statement of the theorem is obtained by observing that the inequality is valid for any $\sigma_{XY}\in P$.
\end{proof}

\end{document}